\documentclass[a4paper]{article}

\usepackage[english]{babel}
\usepackage[utf8]{inputenc}
\usepackage[T1]{fontenc}

\usepackage[a4paper,top=3cm,bottom=2cm,left=3cm,right=3cm,marginparwidth=1.75cm]{geometry}

\PassOptionsToPackage{hyphens}{url}\usepackage{hyperref}
\usepackage{breakurl}
\hypersetup{breaklinks=true}
\usepackage{amsmath}
\usepackage{amssymb}	
\usepackage{amsthm}
\usepackage{tikz-cd}
\usepackage{graphicx}
\usepackage[affil-sl]{authblk}

\newcommand{\fld}{\ensuremath{\Bbbk}}
\newcommand{\dd}{\ensuremath{\mathrm d}}

\newcommand{\dl}[2]{\frac{\partial_L{#1}}{\partial{#2}}}
\newcommand{\dr}[2]{\frac{\partial_R{#1}}{\partial{#2}}}
\newcommand{\dual}[1]{{#1}^*}
\newcommand{\Fun}[1]{\mathcal{F}(#1)}
\newcommand{\Funpos}[1]{\mathcal{F}_{\mathrm{w}+}(#1)}
\newcommand{\hdeg}[1]{{|#1|}}

\newcommand{\dg}[1]{|{#1}|}
\newcommand{\id}{1}
\newcommand{\ot}{\otimes}

\newcommand{\Ber}{\operatorname{Ber}}
\newcommand{\Hzw}{\ensuremath{\mathcal{H}_\mathrm{rel}}}
\renewcommand{\Im}{\operatorname{Im}}
\newcommand{\sus}{{\ensuremath{\uparrow }}}
\newcommand{\des}{{\ensuremath{\downarrow }}}

\newcommand{\Htpy}{\ensuremath{K}}
\newcommand{\htpy}{\ensuremath{k}}
\newcommand{\Hlgy}{\ensuremath{H}}
\newcommand{\Sfree}{\ensuremath{S_\mathrm{free}}}
\newcommand{\Sint}{\ensuremath{S_\mathrm{int}}}

\newtheorem{definition}{Definition}
\newtheorem{theorem}{Theorem}
\newtheorem{lemma}{Lemma}
\newtheorem{remark}{Remark}

\makeatletter
\def\blfootnote{\xdef\@thefnmark{}\@footnotetext}
\makeatother

\title{Quantum $L_\infty$ Algebras and the Homological Perturbation Lemma}
\author{Martin Doubek}
\author{Branislav Jurčo}
\affil{Mathematical Institute, Faculty of Mathematics and Physics, Charles University, Prague 186 75, Czech Republic}
\author{Ján Pulmann\footnote{\href{mailto:Jan.Pulmann@unige.ch}{Jan.Pulmann@unige.ch}   \\}}
\affil{Section of Mathematics, University of Geneva, Switzerland}

\date{}
\begin{document}
\maketitle
\blfootnote{The last two authors are fully responsible for all the mistakes found in this paper.}
\vspace*{0.2cm}
\begin{center}
	\LARGE \textit{Dedicated to the memory of Martin Doubek.}
\end{center}
\vspace*{1cm}

\begin{abstract}
Quantum $L_\infty$ algebras are a generalization of $L_\infty$ algebras with a scalar product and with operations corresponding to higher genus graphs. We construct a minimal model of a given quantum $L_\infty$ algebra via the homological perturbation lemma and show that it's given by a Feynman diagram expansion, computing the effective action in the finite-dimensional Batalin--Vilkovisky formalism. We also construct a homotopy between the original and this effective quantum $L_\infty$ algebra.
\end{abstract}

\section{Introduction}
Quantum $L_\infty$ algebra on a graded vector space $V$ is given by a sequence of symmetric maps $\lambda^g_n: V^{\otimes n} \to V$ and an odd symplectic form $\omega: V\otimes V \to \fld$, satisfying some conditions. The map $\lambda^0_1: V \to V$ squares to zero, so that we can consider its cohomology $\Hlgy$. In this paper, we describe how to transfer the rest of the maps $\lambda^g_n$ to a new quantum $L_\infty$ algebra on $\Hlgy$.
\medskip

One way to transfer a quantum $L_\infty$ algebra is to use the Batalin--Vilkovisky algebra structure on $\Fun{V}$, the space of functions on $V$, which is induced by the odd symplectic form $\omega$. The quantum $L_\infty$ algebra on $V$ can be encoded into an \emph{action} $S\in\Fun{V}$ which solves the \emph{quantum master equation}
$$
\Delta e^{S/\hbar} = 0\,.
$$
Here, $\Delta$ is the Batalin--Vilkovisky Laplacian, a second order differential operator on $\Fun{V}$. 
One can then define the effective action by integrating over the complement of $\Hlgy$ in $V$, obtaining a function on the cohomology $\Hlgy$
$$e^{W/\hbar} = \int_{\Hlgy^{\mathrm{C}}} e^{S/\hbar}\,.$$
It is a simple consequence of the properties of the path integral that the resulting $W$ again solves the quantum master equation on $\Fun{\Hlgy}$. This approach has already been used in a similar context, either by directly defining $W$ as a diagram expansion \cite{BarannikovSolving,Chuang2010,Costello2010, Mnev2008}  or by defining the path integral~\cite{BraunMaunder}.
\medskip

In this paper, we will instead use the homological perturbation lemma, or HPL. Input data to HPL are two vector spaces with a choice of a deformation retract between them. For us, it will be maps
\begin{gather*} 
\begin{tikzcd}[ampersand replacement=\&]
(V,\lambda_1^0) \arrow[shift left=.9ex]{r}{p} \arrow[out=-160, in=160, loop, distance=2em]{}{\htpy} \& (\Hlgy,0) \arrow[shift left=.9ex]{l}{i}
\end{tikzcd}
\end{gather*}
satisfying some conditions. These data induce a similar retract between $\Fun{V}$ and $\Fun{\Hlgy}$. Then, we can interpret the Batalin--Vilkovisky Laplacian $\hbar\Delta$ as a perturbation to the differential $\lambda_0^1$. The HPL then transfers the perturbation to $\Fun{\Hlgy}$ and gives formulas for a new deformation retract. We will show that the perturbed projection map $P_1:\Fun{V}\to \Fun{\Hlgy}$ is given by a path integral and thus can be used to define an effective action.
Moreover, from the HPL one can easily extract an explicit homotopy between the original and the effective action.
\medskip

The homological perturbation lemma was discovered by Brown \cite{BrownLemma}, with similar formulas appearing already in work by Shih \cite{shih1962homologie}. The same result was then later published by Gugenheim \cite{gugenheim1972}, for other notable references see also Huebschmann \cite{HuebschmannLemma} and Lambe, Stasheff \cite{Lambe1987}.
The connection of the HPL and the path integral appears in the literature as well, see Section~\ref{SECRemarks} of this paper for a more detailed review.

Carlo Albert presented a work very similar to this paper at a Cargese conference in 2009 \cite{Albert}. There, he explained that one can see a scalar BV path integral as the HPL, but the work was never published.

In future work, we would like to extend the HPL approach to minimal models of algebras over Feynman transforms of modular operads and over cobar constructions of properads, e.g. the $IBL_\infty$ algebras \cite{MuensterSachsClassification, MuensterSachsQOC}.

\subsection{Organization of the paper}
In Section~\ref{SECBVandQL}, we start by introducing the Batalin--Vilkovisky formalism, serving as a heuristic for the path integral defined by the HPL. Then, we define quantum $L_\infty$ algebras as solutions to the 
quantum master equation.

In Section~\ref{SECMinimalModel}, we recall the homological perturbation lemma and we construct a deformation retract between $\Fun{V}$ and $\Fun{\Hlgy}$.

In Section~\ref{SECTRANSFER}, we apply the HPL to the constructed deformation retract and show that we obtain a quantum $L_\infty$ algebra on $\Fun{\Hlgy}$. We also define a homotopy of solutions of quantum master equation and show that the effective action $W$ is homotopic to the original action $S$.

In Section~\ref{SECRemarks}, we describe the relation of this paper to the mentioned works \cite{BarannikovSolving,Chuang2010,Costello2010, Mnev2008, BraunMaunder} in more detail.

\subsection{Notation and conventions}
For us, the field $\fld$ is always $\mathbb R$ or $\mathbb C$. All of the graded vector spaces are degree-wise finite-dimensional. We use a cohomological convention, with the differential of degree 1.
For $F$ an element of a graded vector space, we denote by $\hdeg{F}$ its degree. The suspension operator is defined by $(\sus V)_{i+1} = V_{i}$, desuspension is given by $(\des V)_{i-1} = V_{i}$ and $(\mathbf{r}V)_i = V_{-i}$. To shorten formulas, we sometimes use the Einstein summation convention.

\section{Batalin--Vilkovisky formalism and quantum $L_\infty$ algebras}
\label{SECBVandQL}
Batalin--Vilkovisky (or BV) formalism \cite{Batalin1981} was developed in quantum field theory as a tool to manipulate ill-defined
path integrals. Later, a geometric interpretation was given by Schwarz \cite{Schwarz1993}. We start this section
by reviewing its properties, which will serve as a heuristic for working with the homological perturbation lemma.

Given a gauge theory, with fields (including ghosts) $\phi^i$, one introduces antifield $\phi^\dagger_i$ for each field and extends the action $S[\phi]$ to $S[\phi, \phi^\dagger]$ such that $S[\phi, \phi^\dagger = 0] = S[\phi]$. The statistics of an antifield is opposite to that of a corresponding field, so one has an odd pairing on the space of fields and antifields. The space of fields is a Lagrangian subspace of this total space. 

The path integral of $\exp(iS[\phi]/\hbar)$ over fields is then generalized to a path integral over any Lagrangian subspace, with a hope that it is more amenable to a perturbative expansion. For the result to make sense, the BV path integral needs to be invariant under (at least small) changes of the Lagrangian subspace. This turns out to be true for $\Delta$-closed integrands, 
where $\Delta$ is a so-called BV Laplacian, defined using the odd pairing
$$\Delta = \pm \frac{\delta_\mathrm R}{\delta \phi^i}\frac{\delta_\mathrm L}{\delta \phi^\dagger_i}\,.$$
This is a second order differential operator which squares to zero. Thus, we will require that the weight $\exp(i S[\phi, \phi^\dagger]/\hbar)$ is $\Delta$-closed, which should be understood as a generalization of a gauge-invariance of $S$.

The BV Laplacian induces a bracket on the space of functionals, defined by a formula
\begin{equation}
\label{EQBVfirst}
\Delta (FG) = (\Delta F) G + (-1)^{\hdeg F} F \Delta G + (-1)^{\hdeg F} \{F, G \}\,.
\end{equation}
A simple calculation using this formula shows that
$$\Delta e^{iS/\hbar} = \frac{i}{\hbar} \left( \Delta S + \frac{i}{2\hbar}\{S, S\}  \right)e^{iS/\hbar}\,,$$
i.e. the condition that $e^{iS/\hbar}$ is $\Delta$-closed can be equivalently stated as
\begin{equation*}
2i\hbar \Delta S - \{S, S\} = 0\,,
\end{equation*}
which is the well-known \emph{quantum master equation}. 

In the following, will drop the factor $i$ in the exponent to simplify formulas, i.e. we will take a weight $e^{S/\hbar}$. Then, the master equation becomes
$$2\hbar\Delta S + \{S, S\} = 0\,.$$
Let us now denote by $V$ the space of fields and antifields and assume that it decomposes into $V = V'\oplus V''$ such that $\Delta$ also decomposes as $\Delta = \Delta' + \Delta''$ (this amounts to $V'$ and $V''$ being symplectic w.r.t. the odd pairing). Then we can integrate out the fields in $V''$ by choosing a Lagrangian subspace $L'' \subset V''$, thus obtaining a functional of the fields $V'$ only. If we apply this to $e^{S/\hbar}$, we can define an \emph{effective action} $W$ by
\begin{equation*}
e^{W/\hbar} \equiv \int_{L''} e^{S/\hbar}\,.
\end{equation*}
Note that this action will depend on the choice of $L''$, since $e^{S/\hbar}$ is not $\Delta''$-closed in general.

This effective action satisfies the master equation in the BV algebra on $\Fun{V'}$, which can be easily proven
\begin{align*}
\Delta' e^{W/\hbar}  = \Delta' \int_{L''}e^{S/\hbar} 
=  \int_{L''}\Delta' e^{S/\hbar} 
=  \int_{L''} (\Delta - \Delta '') e^{S/\hbar} = 0\,.
\end{align*}
Here, we moved $\Delta'$ under the integral because $\Delta'$
and the integral act on different variables. The last equality holds because $\Delta e^{S/\hbar}=0$ by master equation
and $\int_{L''} \Delta'' (\dots)=0$ follows from integration by parts, if $e^{S/\hbar}$ vanishes at infinity.
\medskip

We can also use the path integral to define an \emph{effective observable}, a morphism which takes functionals on $V$ to functionals on $V'$. Let $\Sfree$ be the classical, quadratic part of the action, i.e. a kinetic term, which determines the propagator. We will assume that $\Delta e^{\Sfree/\hbar} = 0$ and define the effective observable as
\begin{equation*}
F \mapsto \int_{L''}  F\,e^{\Sfree/\hbar}\,.
\end{equation*}
For us, it will be important that this is a chain map between $\hbar \Delta + \{\Sfree, -\}$ and $\hbar\Delta'$. This can be demonstrated by
	$$\hbar\Delta' \int_{L''} Fe^{\Sfree/\hbar} =
\int_{L''} \hbar\Delta' \left(Fe^{\Sfree/\hbar} \right)=
\int_{L''} \hbar\Delta \left(Fe^{\Sfree/\hbar}\right)\,.$$
Now we use the fact that for any degree 0 functional $A$, the map $F \mapsto e^{-A/\hbar} \hbar\Delta (F e^{A/\hbar})$ squares to zero. Moreover, if $A$ solves the quantum master equation, we have from Equation \eqref{EQBVfirst} 
$$e^{-A/\hbar} \hbar\Delta (F e^{A/\hbar}) = \hbar \Delta F + e^{-A/\hbar}\hbar\{e^{A/\hbar}, F\} = \hbar\Delta F + \{A, F\}\,,$$
where we used that the bracket is a derivation in both of its arguments. Thus, we get
$$\hbar\Delta' \int_{L''} Fe^{\Sfree/\hbar} = \int_{L''} \left( \hbar\Delta F + \{\Sfree, F\} \right)e^{\Sfree/\hbar}\,.$$
Note that the effective action $e^{W/\hbar}$ can be computed as an effective observable of $e^{(S-\Sfree)/\hbar}$.
\medskip  

We will also use a \emph{normalized effective observable}, which is defined by
$$F \mapsto e^{-W/\hbar} \int_{L''} F e^{S/\hbar}\,.$$
It also intertwines two differentials, this time $\hbar\Delta + \{S, -\}$ and $\hbar\Delta' + \{W, -\}'$. Here, $\{,\}'$ is the BV bracket coming from $\Delta'$.

\subsection{Finite-dimensional BV formalism}
We will now describe the mathematical framework we will use. Instead of the infinite-dimensional space of fields and antifields, we will take, as a model, a $\mathbb{Z}$-graded vector space which is finite-dimensional in every degree.
\begin{definition}
	\label{def:BValgebra}
	A \textbf{BV algebra} is a graded commutative associative algebra on a graded vector space $\mathcal F$
	with a bracket $\{, \}: \mathcal F^{\otimes 2}\to \mathcal F$ of degree 1 that satisfies
	\begin{align}
	\{F, G\} &= -(-1)^{(\hdeg F +1) (\hdeg G +1)} \{G, F\} \,, \nonumber\\
	\label{eq:Jacobi}
	\{ F, \{G, H\}\} &= \{ \{F, G\}, H\} + (-1)^{(\hdeg F+1)(\hdeg G +1) } \{G, \{F, H\}\}\,,  \\
	\{F, GH\} &= \{F, G\}H + (-1)^{(\hdeg F +1)\hdeg G} G \{F ,H  \} \nonumber
	\end{align}
	and a square zero operator $\Delta : \mathcal F \to \mathcal F$ of degree 1 such that 
	\begin{equation}
	\label{eq:BVgood}
	\Delta (FG) = (\Delta F) G + (-1)^{\hdeg{F}} F \Delta G + (-1)^{\hdeg F} \{F, G \}\,.
	\end{equation}
	For algebras with unit $1$, we will require $\Delta(1) = 0$. 
\end{definition}
Since the bracket can be 
defined using $\Delta$, one can define a BV algebra
using only $\Delta$. The Poisson and Jacobi identities of the bracket are then encoded in the so-called seven-term identity, which is a version of Leibniz identity for second-order differential operators
\begin{equation*}
\begin{split}
\Delta(FGH) =& \Delta (FG)H + (-1)^{\hdeg{G}\hdeg{H}} \Delta(FH)G + (-1)^{\hdeg F (\hdeg G+\hdeg H)}\Delta(GH) F \\
& - \Delta(F) GH - (-1)^{\hdeg F} F\Delta(G)H - (-1)^{\hdeg F +\hdeg G}FG\Delta (H)\,.
\end{split}
\end{equation*}
In the following, we will also use a compatibility between $\Delta$ and $\{, \}$ which can be derived from $\Delta^{2}(FG) = 0$
\begin{equation}
\label{eq:delta-and-bracket}
\Delta \{F,G\} = \{\Delta F, G\} + (-1)^{\hdeg F + 1 } \{F, \Delta G\}\,.
\end{equation}
\smallskip

Our main example of a BV algebra will be the algebra of functions on an odd symplectic vector space.
\begin{definition}\label{def:svs}
For a graded vector space $V$, an \textbf{odd symplectic form} of degree $-1$ is a nondegenerate graded-antisymmetric bilinear map $\omega: V\otimes V \to \fld$. A vector space equipped with such form is called an \textbf{odd symplectic vector space}.

If the graded vector space also has a differential $Q$ such that\footnote{This compatibility ensures that the cohomology of $Q$ will inherit a symplectic structure from $V$.}
$$\omega(1\otimes Q + Q \otimes 1 )= 0\,,$$
we call such vector space a \textbf{dg symplectic vector space}.
\end{definition}
To define the space of functions on $V$, we recall the definition of the dual.
\begin{definition}
	For a graded vector space $V$ , the \textbf{graded dual} $\dual{V}$ is defined as $(\dual{V})_i = \dual{(V_{-i})}$.
	\smallskip
	
	Let $f : V \to W$ be a map of graded vector spaces. Its \textbf{transpose} $\dual{f} : \dual{W} \to \dual{V}$ is defined by
	\begin{equation*}
	\dual{f}(\alpha) \equiv (-1)^{\hdeg{f}\hdeg{\alpha}} \alpha\circ f
	\end{equation*}
	for $\alpha \in \dual{W}$. Note that this implies  $\dual{(fg)} = (-1)^{\hdeg f \hdeg g}
	\dual{g}\dual{f}$.
	\smallskip

	The basis $\{\phi^i\}$ of $\dual{V}$ dual to a basis $\{e_i\}$ of $V$ is defined by
	$$\phi^i(e_j) = \delta^i_j\,.$$
\end{definition}
\begin{definition}
The \textbf{space of formal functions} on $V$ is defined as
\begin{equation*}
\Fun{V} \equiv \prod_{g\ge 0, n\ge 0}(\dual{V})^{\odot n}\otimes \fld \hbar^{g} \,.
\end{equation*}
We take a product over all nonnegative symmetric powers of $\dual{V}$ and all non-negative powers of $\hbar$. In other words, we work with formal power series in elements of $\dual{V}$ and in $\hbar$. By convention, $V^{\odot 0} = \fld$. The graded commutative product on $\Fun{V}$ is the $\hbar$-linear extension of the product on the space of symmetric powers of $\dual{V}$.

We will define a BV algebra structure on $\Fun{V}$ in coordinates. Choosing a basis $e_i$ of $V$, we get a
matrix 
\begin{equation*}
\omega_{ij} = \omega(e_i, e_j)\,.
\end{equation*}
The BV Laplacian $\Delta: \Fun{V} \to \Fun V$ is defined using $\omega^{ij}$, the inverse of $\omega_{ij}$, as
\begin{equation*}
\Delta F \equiv \frac 12\sum_{i, j} (-1)^{\hdeg{\phi^i}}\omega^{ij} \frac{\partial^2_L F}{\partial \phi^i \partial \phi^j}\,,
\end{equation*}
where $\phi^i\in \dual V$ is the dual basis of $e_i$. The corresponding bracket is 
\begin{equation*}
\{F,G\} \equiv \sum_{i, j} \frac{\partial_R F}{\partial \phi^i}
\omega^{ij} 
\frac{\partial_L G}{\partial \phi^j}
\,.
\end{equation*}
The partial derivatives are graded and $\hbar$-linear.
\end{definition}
The BV operator has a beautiful geometrical origin, due to Schwarz \cite{Schwarz1993} and Khudaverdian \cite{Khudaverdian}.
There, it is the divergence operator of Hamiltonian vector fields, with respect to some chosen volume form. In our case,
we have a canonical (up to a constant multiple) choice, given by the vector space structure on the graded manifold $V$. Then, the BV operator is defined
by
\begin{equation}\label{EQorigdivergence}
\Delta(F) \dd V= (-1)^{\hdeg F} \frac{1}{2} \mathcal L_{ \{F, -\} } \dd V\,,
\end{equation}
where $\dd V$ is a volume form induced by the coordinates on $V$ (see \cite[Equation~2.1 and Equation~2.7]{Khudaverdian}). 
We will also need the transformation property of $\Delta$ with respect to a symplectic diffeomorphism $\Phi$ 
\begin{equation}\label{EQBVTransformation}
 \Phi_* \circ \hbar \Delta \circ \Phi^* =  \hbar \Delta +\frac{1}{2} \{  \log \Ber (\partial \Phi) , -\} \,,
\end{equation} 
where $\Ber$ is the graded version of determinant \cite[Equation~2.11]{Khudaverdian}.
\medskip

Instead of volume forms, we will use semidensities, which are a more fundamental object. For us, they will be just
objects of the form $F\, \dd^{\frac 12} V$ with $F\in \Fun{V}$, which transform with a factor equal to the square root of the Berezinian.
We will write formally
\begin{equation}\label{EQorigdivergence}
\Delta(F) \dd^{\frac 12} V= (-1)^{\hdeg F}\mathcal L_{ \{F, -\} } \dd^{\frac 12} V\,.
\end{equation}
\begin{remark}
	The transformation property of $\Delta$ can be now seen as a simple compatibility of the Lie derivative $\mathcal L$ 
	with (symplectic) diffeomorphisms. Indeed, applying $\Phi_*\circ\mathcal L_{ \{F, -\} } = \mathcal L_{ \{\Phi_*(F), -\} } \circ \Phi_*$ on $\dd V$, 
	we get exactly Equation~\eqref{EQBVTransformation}
\end{remark}

\subsection{Existence of $e^{S/\hbar}$} \label{SSECExistence}
\begin{remark}
	This subsection explains how to define the quantum master equation rigorously, mainly dealing with issues concerning the powers of $\hbar$. The important parts are the formulas from Lemma~\ref{LEMex1} and Lemma~\ref{LEMex2}, and formulas at the end of this section, the rest is not very enlightening.
\end{remark}
Of course, the exponential $e^{S/\hbar}$ is not an element of $\Fun{V}$, since it contains arbitrary negative powers of $\hbar$. 
\begin{definition}
Allowing all the powers of $\hbar$, we get a space
\begin{equation*}
\mathcal F_{\text{arbitrary}}{(V)} \equiv \prod_{g \in \mathbb Z, n\ge 0}  (\dual{V})^{\odot n} \otimes \fld \hbar^{g} \,.
\end{equation*}
For a homogeneous vector in $(\dual{V})^{\odot n} \otimes \fld \hbar^{g}$, let us call the number $n$ the \textbf{polynomial degree} and the number $g$ the \textbf{genus}.	
\end{definition}
It is not possible to multiply any two elements of $\mathcal F_{\text{arbitrary}}{(V)}$, but we can single out a subspace of elements that are closed under multiplication 
\begin{equation*}
\mathcal F_{\text{finite}}{(V)} \equiv \left\{  v \in \mathcal F_{\text{arbitrary}}{(V)}
\Biggr\rvert \;\;\parbox{20em}{ the component $(v)_n$ of $v$ of polynomial degree $n$ has a lower bound on genus, for each $n$.} \right\}\,.
\end{equation*}
Elements $F, G$ of  $\mathcal F_{\text{finite}}{(V)}$ can be multiplied since, to the polynomial degree $n$ and genus $g$ of $FG$, only a finite number of components of $F$ and $G$ contribute. The BV algebra structure can be defined here by the same formulas as for $\Fun{V}$.
\medskip

To avoid discussing exponentials of constant terms, we will ignore them for now. Denoting the subspace of  $\mathcal F_{\text{finite}}{(V)}$ with no constant part as $\mathcal F_{\text{finite, n.c.}}{(V)}$, 
the exponential  of $A \in \mathcal F_{\text{finite, n.c.}}{(V)}$ is 
$$e^{A} = 1 + A + \frac{1}{2!} A^2 + \dots$$
This exponential (or any power series) is well defined, since only the first $k+1$ terms can contribute to the polynomial degree $k$ of the result. Thus, $e^A$ is finite and we can consider the quantum master equation.
\begin{lemma}\label{LEMex1}
	If $S$ is a degree 0 element of $\mathcal F_{\text{finite, n.c.}}{(V)}$, then 
	\begin{equation*}
	\Delta e^{S/\hbar} = \frac{1}{\hbar^2} e^{S/\hbar} \left( \hbar\Delta S + \frac 12 \{S, S\} \right)\,.
	\end{equation*}
\end{lemma}
\begin{proof}
	It is a simple consequence of Equation~\eqref{eq:BVgood} that 
	$\Delta S^n = n S^{n-1} \Delta S + \frac{n(n-1)}{2} \{S, S\}S^{n-2}$. Thus, for a power series $f(S) = \sum_{n\ge 0} f_n S^n$, we have
	$$\Delta( f(S)) = \sum_{n\ge 0} f_n (n S^{n-1} \Delta S + \frac{n(n-1)}{2} \{S, S\}S^{n-2}) = f'(S)\Delta S + \frac 12 f''(S) \{S, S\} \,.$$
\end{proof}
The next result we will need is the twisting of $\Delta$ by $e^{A/\hbar}$. 
\begin{lemma}\label{LEMex2}
	For $A\in \mathcal F_{\text{finite, n.c.}}{(V)}$ of degree 0 and $F \in \mathcal F_{\text{finite}}{(V)}$, the following identity holds
	\begin{equation}\label{EQTwistDelta}
	e^{-A/\hbar} \hbar\Delta \left( F e^{A/\hbar}\right) = \hbar \Delta F + \{A, F\} + \frac{1}{\hbar} \left( \hbar \Delta A + \frac 12 \{ A, A \}\right) F\,.
	\end{equation}
	Moreover, if we define the \textbf{twisted} BV Laplacian as
	$$T_A(F) \equiv \hbar \Delta F + \{A, F\}\,,$$
	then
	$T_A^2 = 0$ iff $A$ solves the quantum master equation.
\end{lemma}
\begin{proof}
	The first equation is an immediate consequence of Equation~\eqref{eq:BVgood}. The square of $T_A$ can be written as
	$$(T_A)^2(F) = \hbar \Delta \{A, F \} + \{A, \hbar \Delta F + \{A, F\}\} = \{\hbar \Delta A + \frac{1}{2} \{A, A\}, F \}\,,$$
    where we used Equation~\eqref{eq:delta-and-bracket} and the identity $2\{A, \{A, F\}\} =\{ \{A, A\}, F\}$, which follows from the Jacobi identity~\eqref{eq:Jacobi}. However, since $\omega$ is non-degenerate, this means that $\hbar \Delta A + \frac{1}{2} \{A, A\}$ is an odd constant, which can only be 0.
\end{proof}
Note that it is also possible to twist step by step. Take $A, B \in \mathcal F_{\text{finite, n.c.}}{(V)} $ such that $A$ satisfies the master equation.
Then we can twist $\hbar\Delta + \{A, -\}$ by $B$, which will satisfy
$$e^{-B/\hbar}T_A(Fe^{B/\hbar}) = T_{A+B}(F)$$
iff $A+B$ satisfies the quantum master equation
$$ \hbar e^{-B/\hbar}(\hbar \Delta + \{A, - \}) e^{B/\hbar}  = \hbar\Delta B + \{A, B\} + \frac{1}{2}\{B, B\} = 0\,.$$

We finish by introducing the \textbf{weight grading} of Braun and Maunder \cite{BraunMaunder}.
\begin{definition}
	The \textbf{weight} of an element $v\in \mathcal F_{\text{arbitrary}}{(V)}$ of polynomial degree $n$ and genus $g$ is $w=2g + n$. The space $\Funpos{V}$ is defined as
	$$ \Funpos{V} \equiv \prod_{w  \ge 1} \left( \bigoplus_{2g+n = w} \hbar^{g} (\dual{V})^{\odot n} \right) \,, $$
	i.e. elements of positive weight with only finitely many elements of each weight. 
\end{definition}
Since the multiplication is of weight zero, the space $\Funpos{V}$ is closed under multiplication. Moreover, it is also closed under taking arbitrary power series without a constant coefficient.   Note that $\Funpos{V}$ contains components of polynomial degree 0.
\medskip

The weight grading is useful because it is preserved by $\hbar \Delta$ and consequently also by the path integral. In other words, we will show that a path integral of an element of $\Funpos{V}$ is again a well-defined element here, and it makes sense to talk about its logarithm.
\medskip

Equipped with these notions, we can put some conditions on the action $S\in\Fun{V}$. Let us decompose it to the part of polynomial degree 2 and genus 0, called $\Sfree$, and the rest $\Sint = S - \Sfree$. The part $\Sfree$ has weight 2. \textbf{In the following, we will assume that $\Sint/\hbar$ is an element of $\Funpos{V}$, i.e. $\Sint$ is in weight $3$ and more.} For $\Sint \in \Fun{V}$, this means it starts in polynomial degree 3 for genus 0 and in polynomial degree 1 in genus 1. Since the constant part of $\Sint/\hbar$ is in weight 1 or more, all the expressions in Lemma~\ref{LEMex1} and Lemma~\ref{LEMex2} are well defined and we can apply the lemmas to $S$. Thus, we have the master equation for $\Sfree + \Sint$
$$\hbar \Delta (\Sfree + \Sint) + \frac 12 \{\Sfree + \Sint, \Sfree + \Sint\} = 0\,.$$
If $S$ solves the quantum master equation, then the quadratic, genus $0$ part of the quantum master equation is just $\{\Sfree, \Sfree \}= 0$, which means that $\{\Sfree, - \}$ squares to 0. Moreover, since $\Sfree$ is of degree 0, $\Delta \Sfree$ is a constant of degree 1, i.e. zero. Thus, $\Sfree$ is also a solution of the master equation. Following the remark after Lemma~\ref{LEMex2}, this means that we have a differential
$$T_{\Sfree} = \hbar \Delta + \{\Sfree, -\}\,,$$
which can be twisted to the full differential
$$e^{-\Sint/\hbar}T_{\Sfree}(F e^{\Sint/\hbar}) = \hbar \Delta F + \{\Sfree + \Sint, F\}\,.$$
The master equation then reduces to
\begin{equation}\label{EQMasterEqSfree}
T_{\Sfree}(e^{\Sint/\hbar}) = e^{-\Sfree/\hbar} \hbar \Delta e^{(\Sfree + \Sint)/\hbar} = 0 \,,
\end{equation}
or equivalently,
\begin{equation*}
\hbar \Delta \Sint + \{\Sfree, \Sint \} + \frac{1}{2} \{ \Sint, \Sint \}=0\,.
\end{equation*}

\subsection{Quantum master equation and quantum $L_\infty$ algebras}
The first appearance of a quantum $L_\infty$ algebra was in the correlation functions of the closed string field theory of Zwiebach 
\cite{zwiebach-closed}.  He defined \emph{string functions}, graded symmetric multilinear maps from the relevant Hilbert space $\Hzw$ into $\mathbb C$
and proved that they satisfy a series of identities \cite[Equation~4.10]{zwiebach-closed}, called the \emph{main identity}.
Then, he showed that this identity is equivalent to the quantum master equation (see \cite[Section~4.4]{zwiebach-closed})
for an action $S$ encoding all the string functions.
\medskip

The main identity generalizes the defining relations of a cyclic $L_\infty$ algebra. This was elaborated on by Markl in \cite{markl2001loop},
where he defined loop homotopy Lie algebras as maps satisfying the main identity. In this work, he showed that loop homotopy Lie algebras (or, as we will call them, quantum $L_\infty$ algebras) can be viewed as algebras over the Feynman transform of the modular operad $\operatorname{Mod}(Com)$, generalizing the cobar construction of $L_\infty$ algebras. 

We will define a quantum $L_\infty$ algebra to be a solution of the quantum master equation. To get an action of degree $0$, we shift 
$\Hzw$ twice. Then, the inner product $\langle, \rangle$ on $\Hzw$ defines a degree $-1$ symplectic form  on $V= \des\des \Hzw$ by
$$\omega(v_1, v_2) \equiv (-1)^\hdeg{v_1} \langle \sus\sus v_1, \sus\sus v_2 \rangle\,.$$
\begin{definition}
	A \textbf{quantum $L_\infty$ algebra}, or a \textbf{loop homotopy Lie algebra}, on a symplectic vector space $(V, \omega)$, is given by a degree 0 element $S\in \Fun{V}$ that satisfies the quantum master equation
	\[2\hbar\Delta S + \{S, S\} = 0\,.\]
	 Moreover, we require that the genus 0 part of $S$ is at least quadratic and genus 1 part is at least linear.
\end{definition}
A definition of a quantum $L_\infty$ algebras as solutions to the master equation appeared in the work of Braun and Lazarev \cite[Section~6]{BraunLazarev}.

Any quantum $L_\infty$ algebra comes with a differential on $V$ compatible with the symplectic structure, corresponding
to the differential $\{ \Sfree, - \}$ (see also Lemma~\ref{LEMQSfree}).
\begin{lemma}
	Given a quantum $L_\infty$ algebra, let us decompose the action $S$ as
	\begin{equation*}
	S = \sum_{n\ge 2, g \ge 0} \hbar^g \frac{s_{n}^g}{n!} \in \Fun{V}\,.
	\end{equation*}
	with $s_n^{g}$ a graded symmetric function taking $n$ vectors from $V$.
	Then, the map $Q\colon V\to V$ defined by
	\[ s_2^0(v_0, v_1) = (-1)^{\hdeg{v_0}} \omega(v_0, Qv_1)\,, \]
	gives a dg symplectic vector space structure  $(V, \omega, Q)$  (see Definition~\ref{def:svs}).
\end{lemma}
\begin{proof}
	The fact that $Q^2 = 0$ comes from looking at genus $0$, quadratic part of the quantum master equation. The compatibility 
	of $Q$ with $\omega$ follows from the symmetry of $s_2^0$.
\end{proof}

\section{Minimal model}
\label{SECMinimalModel}
Similarly to the definition of $Q$, one can define maps $\lambda^0_n : (V)^{\odot n} \to V $ by \[s^0_{n+1} (v_0, \dots, v_n) = (-1)^{\hdeg{v_0}} \omega(v_0, \lambda^0_n(v_1, \dots, v_n))\,.\]
After a suitable shift (see \cite{markl2001loop}), $\lambda_2^0$ becomes an antisymmetric bracket whose failure to satisfy the Jacobi identity is equal to $ Q \circ \lambda_3^0 - \lambda_3^0\circ(1\otimes 1\otimes Q + 1\otimes Q \otimes 1 + Q\otimes 1 \otimes 1)$.
Moreover, $\lambda_2^0$ is compatible with the differential $Q$, and thus 
descends to a Lie bracket on $\Hlgy$, the cohomology of $V$ w.r.t. $Q$.
\medskip

The task of finding a \emph{minimal model} is to encode the higher operations from $V$ to $\Hlgy$ as well, introducing a quantum $L_\infty$ algebra on $\Hlgy$ compatible with the one on $V$. This makes sense because, thanks to the compatibility $\omega(1\otimes Q + Q \otimes 1 )= 0$, the cohomology $\Hlgy$ inherits a symplectic structure and thus we have a BV algebra structure on $\Fun{\Hlgy}$. 
Minimal model of a quantum $L_\infty$ algebra is therefore given by an action $W\in \Fun{\Hlgy}$ satisfying the quantum master equation $2\hbar \Delta' W + \{W, W\}' = 0$. 

For $L_\infty$ algebras, one requires that there is a quasi-isomorphism connecting the original algebra and the minimal model. In our case, we also have the odd symplectic structure, but requiring that we obtain a symplectomorphism is a very restrictive notion (this is what Kajiura defines as a minimal model \cite[Definition~2.13]{Kajiura}). We give partial answers in Section~\ref{SSECHomotopies} and Section~\ref{SSECMorphism}, using the notion of homotopy of solution of the quantum master equation.
\subsection{Homological perturbation lemma}
Our aim is to define a path integral using the homological perturbation lemma, or HPL. We start by
reviewing HPL, the standard reference is a paper by Crainic \cite{Crainic2004}.
\begin{definition} \label{DEFSS}
	A \textbf{standard situation (SS)} is a pair $(V,Q)$ and $(W,e)$ of dg vector spaces, a pair $p$ and $i$ of their morphisms and a homotopy $\htpy$ between $ip$ and $1_V$
	$$\begin{tikzcd}
	(V,Q) \arrow[shift left=.9ex]{r}{p} \arrow[out=-160, in=160, loop, distance=2em]{}{\htpy} & (W,e) \arrow[shift left=.9ex]{l}{i}
	\end{tikzcd}$$
	that satisfy the following:
	\begin{gather*}
	Q^2=0, \quad e^2=0, \quad \dg{Q}=\dg{e}=1, \\
	pQ=ep, \quad \dg{p}=0,\\
	ie=Qi, \quad \dg{i}=0,\\
	ip-\id_V=Q\htpy+\htpy Q, \quad \dg{\htpy}=-1.
	\end{gather*}
\end{definition}

\begin{definition}
	A \textbf{deformation retract (DR)} is an SS such that 
	$$pi=\id_W.$$
\end{definition}

\begin{definition}
	A \textbf{special deformation retract (SDR)} is a DR such that the following \textbf{annihilation conditions} are met:
	\begin{gather*}
	p\htpy=0, \quad \htpy i=0, \quad \htpy^2=0.
	\end{gather*}
	With this conditions, $ip$, $-\htpy Q$ and $-Q\htpy$ are three projectors such the direct sum of their images gives the whole space $V$.
\end{definition}

If we have a standard situation, we can perturb the differential on $V$ to a new one, requiring that it still squares to zero. The perturbation lemma then gives explicit formulas for a new perturbed standard situation.
\begin{theorem}[Perturbation lemma] \label{THMPL}
	Consider an SS as above:
	\begin{gather} \label{EQSSInPertLemma}
	\begin{tikzcd}[ampersand replacement=\&]
	(V,Q) \arrow[shift left=.9ex]{r}{p} \arrow[out=-160, in=160, loop, distance=2em]{}{\htpy} \& (W,e) \arrow[shift left=.9ex]{l}{i}
	\end{tikzcd}
	\end{gather}
	A \textbf{perturbation} $\delta:V\to V$ of the differential $Q$ is a linear degree $1$ map such that 
	$$(Q+\delta)^2=0.$$
	Equivalently,
	$$\delta^2+\delta Q+Q\delta=0.$$
	Let $\delta$ be a perturbation of $Q$ which is small in the sense that
	$$(1-\delta \htpy)^{-1} \equiv \sum_{i=0}^\infty (\delta \htpy)^i$$
	is a well defined linear map $V\to V$.
	
	Denote 
	\begin{gather}
	Q' \equiv Q+\delta, \nonumber \\
	e' \equiv e+p(1-\delta \htpy)^{-1}\delta i = e+p\delta(1-\htpy\delta)^{-1}i, \nonumber \\
	p' \equiv p+p(1-\delta \htpy)^{-1}\delta \htpy = p(1-\delta \htpy)^{-1}, \nonumber \\
	i' \equiv i+\htpy(1-\delta \htpy)^{-1}\delta i = (1-\htpy\delta)^{-1}i, \nonumber \\
	\htpy' \equiv \htpy+\htpy(1-\delta \htpy)^{-1}\delta \htpy = \htpy(1-\delta \htpy)^{-1}, \nonumber \\
	\begin{tikzcd}[ampersand replacement=\&]
	(V,Q') \arrow[shift left=.9ex]{r}{p'} \arrow[out=-160, in=160, loop, distance=2em]{}{\htpy'} \& (W,e') \arrow[shift left=.9ex]{l}{i'}\,.
	\end{tikzcd} \label{EQSSPerturbedInPertLemma}
	\end{gather}
	Then:
	\begin{enumerate}
		\item \eqref{EQSSPerturbedInPertLemma} is an SS.
		\item If $p$ is a quasi-isomorphism (equivalently: $p$ induces surjective map on cohomology, or $i$ is a quasi-isomorphism, or $i$ induces an injective map on cohomology), then $p'$ is a quasi-isomorphism (equivalently: $p'$ induces surjective map on cohomology, or $i'$ is a quasi-isomorphism, or $i'$ induces an injective map on cohomology).
		\item If \eqref{EQSSInPertLemma} is a DR, then \eqref{EQSSPerturbedInPertLemma} is a DR iff
		$$p(A\htpy^2A+A\htpy+\htpy A)i=0\,,$$
		where $A\equiv(1-\delta \htpy)^{-1}\delta$.
		\item If \eqref{EQSSInPertLemma} is an SDR, then \eqref{EQSSPerturbedInPertLemma} is an SDR.
	\end{enumerate}
\end{theorem}

\begin{proof}
	See \cite{Crainic2004}.
\end{proof}
\subsection{Hodge decomposition}
To construct a special deformation retract between $\Fun{V}$ and $\Fun{\Hlgy}$, we start with a decomposition of the vector space $V$ compatible with the odd symplectic structure.
\begin{lemma} \label{LEMSDR}
	Let $(V,Q,\omega)$ be a dg symplectic vector space.
	Then, there is a decomposition $V \cong H\oplus \Im Q \oplus C$ and maps
	$k, i, p$ such that
	$$\begin{tikzcd}
	( \Hlgy \oplus \Im Q \oplus C \arrow[start anchor={[xshift=1ex]north}, end anchor={[xshift=-3ex]north east}, out=90, in=90, distance=1em]{}{\htpy} \arrow[start anchor={[xshift=-3ex]south east}, end anchor={[xshift=1ex]south}, out=-90, in=-90, distance=1em]{}[below]{Q} ) \arrow[shift left=.9ex]{r}{p} & (\Hlgy,0) \arrow[shift left=.9ex]{l}{i}
	\end{tikzcd}\vspace*{-5mm}$$
	is an SDR.
	
	Here, the space $H$ is isomorphic to the cohomology of $V$ and is symplectic. The subspaces $C$ and $\Im Q$ are both Lagrangian subspaces of the symplectic subspace $\Im Q \oplus C$, isomorphic via $Q$ and $k = Q^{-1}|_{\Im Q}$. The maps $p$ and $i$ are projections and inclusions of the subspace $H$.
\end{lemma}
\begin{proof}
	Such decomposition is called \emph{harmonious Hodge decomposition} by Chuang and Lazarev, who prove its existence
	in \cite[Proposition~2.5, Theorem~2.7]{ChuangLazarevDecomposition}.
\end{proof}
Let us choose bases of these subspaces: $\{a_i\}$ for $\Hlgy$, $\{b_j\}$ for $\Im Q$ and $\{c_k\}$ for $C$. The differential
$Q$ thus takes $c$ to $b$ and is equal to zero on $a$ or $b$. 
The dual vector space $\dual{V}$ is also decomposed into $\dual{\Hlgy}\oplus\dual{(\Im Q)}\oplus\dual{C}$, which have bases $\alpha^i$, $\beta^i$ and $\gamma^i$. The dual of the differential $Q$ then takes $\beta$ to $\gamma$. 

In this basis, we also have a decomposition of the BV algebra structure on $\Fun{V}$.
\begin{lemma}
	Given a decomposition of $V$ as in Lemma~\ref{LEMSDR}, the BV Laplacian on $\Fun{V}$ decomposes as
	 $\Delta = \Delta' + \Delta''$, where 
$$\Delta' = \frac 12\sum_{i, j} (-1)^{\hdeg{\alpha^i}}(\omega')^{ij} \frac{\partial^2_L }{\partial \alpha^i \partial \alpha^j}, \quad\quad \Delta'' = \sum_{i, j} (-1)^{\hdeg{\gamma^i}}(\omega'')^{ij} \frac{\partial^2_L }{\partial \gamma^i \partial \beta^j}\,,$$
where $\omega'_{ij}(\omega')^{jk} = \delta^k_i$ and similarly for $\omega''$. 

The bracket also decomposes as  $\{,\} = \{,\} ' + \{,\} ''$, with analogous formulas.
\end{lemma}
\begin{proof}
	In the basis $(\{a_i\}, \{b_j\}, \{c_k\})$, the symplectic form $\omega$ decomposes to $\omega'_{ij} = \omega(a_i, a_j)$ and $\omega''_{ij} = \omega(b_i, c_j)$. The matrices for $\omega$ and its inverse then look like
	$$\omega = \left(
	\begin{matrix}
	\omega' & 0 & 0 \\
	0 & 0 & \omega'' \\
	0 & -\omega'' & 0
	\end{matrix}
	\right)\,,
	\hspace{3cm}
	\omega^{-1} = \left(
	\begin{matrix}
	(\omega')^{-1} & 0 & 0 \\
	0 & 0 & -(\omega'')^{-1} \\
	0 & (\omega'')^{-1} & 0
	\end{matrix}
	\right)\,.
	$$
\end{proof}
\subsection{General setting}
Now we would like to extend this SDR on $V$ to an SDR between $\Funpos{V}$ and $\Funpos{\Hlgy}$. There are now \emph{two} closely related differentials on $\Funpos{V}$: the one induced by $Q$, and the bracket $\{\Sfree, - \}$. In the end, we want to use a dg vector space $(\Funpos{V}, \{\Sfree, - \})$, so we need to show that $\{\Sfree , - \}$ is compatible with the choice of decomposition $V=\Hlgy\oplus \Im Q\oplus C$.
\begin{lemma}\label{LEMQSfree}
	The differential $\{\Sfree, - \}$, restricted to a map $\dual{V} \to \dual{V}$, is equal to 
	\begin{equation*}
	\{\Sfree, \phi^i\} = - \phi^i \circ Q\,,
	\end{equation*}
	therefore it is an isomorphism $\dual{(\Im Q)} \to \dual{C}$ and restricts to zero on $\dual{C}$ and $\dual{H}$. 
\end{lemma}
\begin{proof}
	Let us evaluate
	$$
	\{\Sfree, \phi^i\}(e_k) = \left( \frac{1}{2} \dr{{s}^0_2}{\phi^a} \omega^{ab} \dl{\phi^i}{\phi^b} \right) (e_k) = (-1)^{\hdeg{e_k} \hdeg{e_a} + \hdeg{e_a}} \omega(e_a, Q (e_k))\omega^{ai}\,.
	$$
	Since $s_2^0$ is of degree 0, we have $\hdeg{e_k} = - \hdeg{e_a}$ and the sign disappears. 
$$
	\{\Sfree, \phi^i\}(e_k) =  \omega(e_a, Q (e_k))\omega^{ai}
	= \omega(e_a, e_m)\omega^{ai} \phi^m(Q (e_k)) = - \phi^i (Q (e_k))
	\,.
$$
Note that $s_2^0 = - \omega''_{ki}Q^k_j \gamma^j\gamma^i$, where $\omega''_{ki} = \omega(b_k, c_i)$
\end{proof}
Using this formula, we can define a homotopy $\Htpy$ on $\Funpos{V}$ as the inverse to $\{\Sfree, - \}$ on $\dual{V}$, extended by a (normalized) Leibniz rule. 
\begin{lemma} \label{LEMSDRReady}
	Given a decomposition as in Lemma~\ref{LEMSDR}, there is a deformation retract 
	$$\begin{tikzcd}
	(\Funpos{V}, \{\Sfree, - \}) \arrow[shift left=.9ex]{r}{P} \arrow[out=-160, in=160, loop, distance=3em, shift left=2.1em]{}{\Htpy} & (\Funpos{\Hlgy},0) \arrow[shift left=.9ex]{l}{I}
	\end{tikzcd}$$
	such that in a basis where $Q(c_i) = Q_i^k b_k$,
	\begin{align}
	\{\Sfree, - \} &= -  \gamma^i Q^k_i \dl{}{\beta^k}, \nonumber \\
	\Htpy(x) &= \frac{1}{\#_{\beta+ \gamma}}\beta^k (Q^{-1})^{i}_k  \dl{x}{\gamma^i}, \quad x\in\Funpos{V}, \label{EQDefDHIP} \\
	I &= \sum_{n\geq 0} (\dual{p})^{\ot n}, \quad P = \sum_{n\geq 0} (\dual{i})^{\ot n}, \nonumber
	\end{align}
	Here, $\alpha$, $\beta$ and $\gamma$ are bases of $\dual{\Hlgy}$, $\dual{(\Im Q)}$ and $\dual{C}$ and the symbol $\#_{\beta+ \gamma}$ denotes, for a monomial $x$, the number of occurrences of variables $\beta^i$ and $\gamma^i$ in $x$. When the number is zero, the operator $\Htpy$ is defined to be zero. The projection $P$ and inclusion $I$ are identities on constant (polynomial weight zero) terms and $\dual{i}$, $\dual{p}$ are duals of $i$, $p$. Explicitly,
	\begin{gather*}
	\dual{i}(\alpha^i) = \alpha^i, \quad \dual{i}(\beta^j)=0=\dual{i}(\gamma^k), \\
	\dual{p}(\alpha^i) = \alpha^i.
	\end{gather*}
\end{lemma}
\begin{proof}
	The only non-trivial identity to check $IP - 1 = [\{\Sfree, - \}, \Htpy]$. For simplicity, let us choose a basis where $-Q^k_i$ is the identity matrix and denote by $\Htpy_0  = \#_{\beta+\gamma} \Htpy=  - \beta^k  \dl{}{\gamma^k}$ the unnormalized homotopy operator. Let us compute
\begin{align*}
	[\{\Sfree, - \}, {\Htpy}_0] &= \left[ \gamma^i \dl{}{\beta^i},  - \beta^k   \dl{}{\gamma^k}\right] \\ &= - \gamma^i \dl{}{\gamma^i} - (-1)^{\hdeg{\beta^i} + \hdeg{\gamma^i}\hdeg{\beta^k} + 1 + \hdeg{\gamma^i}\hdeg{\gamma^k}} \beta^i \delta^k_i  \dl{}{\beta^k} = - \gamma^i \dl{}{\gamma^i} - \beta^k  \dl{}{\beta^k}\,.
\end{align*} 
	This operator, applied on a monomial, will multiply it by minus the number of variables $\gamma$ and $\beta$. Since $\{\Sfree, -\}$ commutes with $\#_{\beta + \gamma}$, the commutator $[\{\Sfree, - \}, \Htpy]$ is then minus the identity on monomials with $\#_{\beta+\gamma}\neq 0$ and zero otherwise. This is, however, exactly $IP-1$.  
\end{proof}
\begin{remark}
	Given an SDR as in~\eqref{EQSSInPertLemma}, the process of inducing an SDR on tensor powers is often called the \emph{tensor trick}, and goes back at least to Eilenberg and Mac Lane \cite[Section 12]{EM}. Their formula for the homotopy is, after symmetrization,
	$$
		\Htpy \equiv \sum_{n\geq 1} \sum_{i=1}^n \frac{1}{n!}\sum_{\sigma\in\Sigma_n} \sigma\cdot\left(\id^{\otimes (n-i)}\ot \dual{\htpy}\ot(ip)^{\ot (i-1)}\right)\,,
	$$
	and since $ip$ is a projector, this gives $\dual{\htpy}$ extended as a derivative, up to a multiplicative factor. One can then check that this factor is equal to $1/\#_{\beta+\gamma}$. To get a homotopy for the differential $\{\Sfree, - \}$, one needs to introduce a sign as in Lemma~\ref{LEMQSfree}. 
	
	Let us also remark that an analogous retract can be defined on $\Fun{V}$ and $\Fun{\Hlgy}$ by the same formulas.
\end{remark}
\section{Transfer}
\label{SECTRANSFER}
\subsection{The two perturbations}
\label{SSEC:twoperturbations}
Recall that we decomposed an action $S\in\Fun{V}$ as
$$S=\Sfree+\Sint,$$
where $\Sfree$ is concentrated in genus $0$  and quadratic in variables of $\dual{V}$, while $\Sint$ is at least cubic in $\dual{V}$ in genus $0$, linear in genus $1$, and there are no restrictions in higher genera.
Since $S$ satisfies the quantum master equation, we have a differential
\begin{equation*}
T_S =  \{ \Sfree, - \} + \hbar \Delta + \{\Sint, -\}\,.
\end{equation*}
Consider the SDR of Lemma~\ref{LEMSDRReady}
\begin{gather} \label{EQHTPOnceAgain}
\begin{tikzcd}[ampersand replacement=\&]
(\Funpos{V},\{\Sfree, -\}) \arrow[shift left=.9ex]{r}{P} \arrow[out=-160, in=160, loop, distance=3em, shift left=2.1em]{}{\Htpy} \& (\Funpos{\Hlgy},0) \arrow[shift left=.9ex]{l}{I}
\end{tikzcd}
\end{gather}
There are two perturbations of $\{\Sfree, -\}$ we will consider:
\begin{itemize}
	\item A perturbation $\delta_1 \equiv \hbar \Delta$. The perturbed differential squares to zero since $S_0$ solves the quantum master equation -- see Section~\ref{SSECExistence} for details. This perturbation will correspond to the unnormalized path integral.
	\item The perturbation $\delta_2 \equiv \hbar \Delta + \{\Sint, -\}$. This perturbation corresponds to the normalized path integral, with weight $S$. 
\end{itemize}
\subsubsection{Perturbation by $\hbar \Delta$}
Consider the SDR~\eqref{EQHTPOnceAgain} and take 
$$\delta_1\equiv\hbar\Delta,$$
as a perturbation. Let's denote the corresponding perturbed maps with subscript 1, e.g.
\begin{equation*}
E_1 = P(1-\delta_1\Htpy)^{-1}\delta_1 I = P(1-\hbar\Delta \Htpy)^{-1}\hbar\Delta I = \hbar P\Delta I = \hbar \Delta',
\end{equation*}
since $\Htpy\Delta I=0$, which follows easily from Equation~\eqref{EQDefDHIP} for $\Htpy$.
The other maps are
\begin{align*}
\Htpy_1 &= \Htpy + \Htpy \hbar\Delta \Htpy + \Htpy \hbar \Delta \Htpy \hbar \Delta \Htpy + \dots = \Htpy + \Htpy \hbar\Delta'' \Htpy + \Htpy \hbar \Delta'' \Htpy \hbar \Delta'' \Htpy + \dots\,, \\
I_1 &= I + \Htpy \hbar\Delta I + \Htpy \hbar\Delta \Htpy \hbar\Delta I + \dots = I\,,\\
P_1 &= P + P \hbar\Delta \Htpy + P \hbar\Delta \Htpy \hbar\Delta \Htpy  + \dots\,.
\end{align*}
where the simplification in $\Htpy_1$ is because $\Delta'$ anticommutes with $\Htpy$ and $\Htpy^2 = 0$. All these maps are of weight 0 and the series converge since $\Delta$ always decreases the polynomial degree by 2. 

\begin{definition}\label{DEFEAPI}
	The \textbf{effective action} $W\in\Funpos{\Hlgy}$ is defined by
	$$e^{W/\hbar} \equiv P_1(e^{\Sint/\hbar}) = P(1-\hbar\Delta \Htpy)^{-1}e^{\Sint/\hbar}.$$
	The \textbf{path integral} is a map $Z:\Funpos{V}\to\Funpos{H}$ defined by
	$$Z(f) \equiv (P_1(e^{\Sint/\hbar}))^{-1}P_1(e^{\Sint/\hbar}f) = e^{-W/\hbar}P(1-\hbar\Delta \Htpy)^{-1}(e^{\Sint/\hbar}f).$$
\end{definition}
\begin{remark}
	Here, we have the issue of the constant $1$ in the expansion of $\exp(X) = 1 + X + \dots$. In the definition of $W$, 1 is annihilated by everything but the first term in $P_1$,
	i.e. $P_1(e^{\Sint/\hbar})$ starts with $1$, and we can take the logarithm. 
	
	For the definition of the path integral $Z(f)$, if we take $f\in \Funpos{V}$, then also $fe^{\Sint/\hbar}\in \Funpos{V}$ and it also makes sense to multiply by the inverse of $P_1(e^{\Sint/\hbar})$,
	again because $P_1(e^{\Sint/\hbar})$ starts with $1$. Thus, $Z(f)$ is again in $\Funpos{V}$
\end{remark}
\begin{theorem} \label{THMCertainComputationII}
	The effective action $W$ is an element of $\Fun{H}$, i.e. it contains only
	nonnegative powers of $\hbar$. Moreover, 
	$W$ satisfies the master equation on $\Fun{H}$:
	$$\hbar \Delta ' W+\frac{1}{2}\{W,W\}'=0\,.$$
\end{theorem}
	
\begin{proof}
	The first part is proven by expressing $P_1(F)$, for $F\in \Fun{V}$, as a Feynman expansion. We begin by noting that every $\Htpy$ in the expansion of $P_1$ adds one variable $\beta$. Since $\Delta$ can remove at most one $\beta$ and the leftmost $P$ is zero on anything with $\beta$, the only nonzero terms of $P_1(F)$ are those where
	\begin{itemize}
		\item $F$ itself has no variables $\beta$ and 
		\item all $\Delta$ remove one $\beta$, i.e. only terms with $\Delta''$.
	\end{itemize}
	We can thus write
	\begin{equation*}
	P_1 = P + P \hbar\Delta'' \Htpy + P \hbar\Delta'' \Htpy \hbar\Delta'' \Htpy  + \dots
	\end{equation*}
	Now, let this act on a monomial with zero variables $\beta$ and $2n$ variables $\gamma$. Each term $\Delta'' \Htpy$ removes two $\gamma$s, so the total numerical factor coming from the normalization of $\Htpy$ is equal to
	$$\frac{1}{2n}\frac{1}{2n-2} \dots \frac{1}{2} = \frac{1}{2^n n!}\,.$$
	Since $\Delta''$ must always remove $\beta$ in order to have nonzero contribution, in $P_1$ we get a repeated application of quadratic differential operator
	$$ \partial_P \equiv \left[  (-1)^{\hdeg{\gamma^i}}(\omega'')^{ij} \frac{\partial^2_L }{\partial \gamma^i \partial \beta^j}, \beta^k (Q^{-1})^{l}_k  \dl{}{\gamma^l} \right] 
	= (-1)^\hdeg{\gamma^i} (\omega'')^{ij} (Q^{-1})^l_j \frac{\partial^2_L }{\partial \gamma^i \partial \gamma^l}\,. $$
	Together with the normalization, we see that we can write $P_1$ as
	$$P_1 = P \exp{(\frac{1}{2}\hbar \partial_P)}\,,$$
	which is by standard arguments a sum over graphs, ending with legs with variables $\alpha$ due to the projection $P$ (see e.g. Lemma 3.4.1 of \cite[Chapter~2]{Costello2010}).
	
	The effective action
	$$W = \hbar \log P [ \exp{(\frac 12 \hbar \partial_P)} \exp{(\Sint/\hbar)}] $$
	thus contains, by Lemma~3.4.1 of loc. cit, only nonnegative powers of $\hbar$.
	\medskip
	
	To show that $W$ is a solution to the quantum master equation, we use that the perturbed map $P_1$ is again a chain map
	$$P_1(\{\Sfree, -\}+\hbar\Delta)=E_1P_1 = \hbar\Delta' P_1\,,$$
	and evaluate this on $e^{\Sint/\hbar}$. Using Equation~\eqref{EQMasterEqSfree}, we get that the left hand side is zero, while the right hand side is equal to
	$\hbar \Delta' e^{W/\hbar}$. 
\end{proof}
\subsection{Perturbation by $\hbar\Delta + \{\Sint, - \}$}
We defined the map $Z$, the normalized path integral, as a map $\Funpos{V} \to \Funpos{H}$. We want to show that it's also a map $\Fun{V} \to \Fun{H}$ and relate it to the perturbation lemma. To do this, we consider the other perturbation from Section~\ref{SSEC:twoperturbations}
$$\delta_2 = \hbar \Delta + \{\Sint, -\} \,.$$
The perturbation lemma then gives the following maps
\begin{align*}
\Htpy_2 &= \Htpy + \Htpy \delta_2 \Htpy + \Htpy \delta_2 \Htpy \delta_2 \Htpy + \dots \,, \\
I_2 &= I + \Htpy \delta_2 I + \Htpy \delta_2 \Htpy \delta_2 I + \dots \,,\\
P_2 &= P + P \delta_2 \Htpy + P \delta_2 \Htpy \delta_2\Htpy  + \dots\,, \\
E_2 &= P \delta_2 I + P \delta_2 \Htpy \delta_2 I  + P \delta_2 \Htpy \delta_2 \Htpy \delta_2 I + \dots \,.
\end{align*}
Here, $\delta_2 = \hbar\Delta + \{\Sint, - \}$ never decreases the weight. To see that the series converge,
note first that any of the above, applied on monomial $x$, will give a finite contribution to any fixed weight. Because for a general element
$F \in \Funpos{V}$, there are only finitely many elements of weight smaller or equal to some number, the perturbed operators are well defined.

A similar argument works when we take $F\in \Fun{V}$.

\begin{theorem} \label{LEMEqualityOfProjections}
	The map $Z$ from Definition~\ref{DEFEAPI} is equal to $P_2$, i.e. 
	$$Z(f)\equiv e^{-W/\hbar} P_1(f e^{S/\hbar}) =P_2(f)\,.$$
	Thus, considering the perturbation $\delta_2$ of a deformation retract {taken on $\Fun{V}$ and $\Fun{H}$ instead of $\Funpos{V}$ and $\Funpos{\Hlgy}$}, we get that 
	$Z$ is a map $\Fun{V}\to\Fun{H}$.
\end{theorem}
To prove this theorem, will need two simple results. 
\begin{lemma} \label{LEMClaimOne}
	$Z(f)=0$ and $P_1(f) = 0$ if $f$ is a monomial with at least one $\beta$.
\end{lemma}

\begin{proof}
	We used this fact already in the proof of Theorem~\ref{THMCertainComputationII}:
	Observe that every nonzero monomial of $\Delta \Htpy(x)$ has at least as many $\beta$'s as $x$ for arbitrary monomial $x\in\Fun{V}$.
	Since $e^{\Sint/\hbar}f$ has at least one $\beta$, then so does $(1-\hbar\Delta \Htpy)^{-1}(e^{\Sint/\hbar}f)=\sum_{n\geq 0}(\hbar\Delta \Htpy)^n(e^{\Sint/\hbar}f)$, and hence vanishes, because $\beta$'s are killed by $P$. 
	
	The proof of $P_1(f) = 0$ is completely analogous.
\end{proof}

\begin{lemma} \label{LEMClaimTwo}
	$Z(Ig)=g$ whenever $g\in\Fun{H}$.
\end{lemma}

\begin{proof}
	Again, every nonzero monomial of $\Delta \Htpy(x)$ has at least as many $\beta$'s as $x$ for arbitrary monomial $x\in \Fun{V}$.
	The nonzero monomials of $P(\Delta \Htpy)^n(x)$ are only those where every $\beta$ added by $\Htpy$ is removed by some $\Delta$.
	Since the number of $\Htpy$'s and $\Delta$'s are equal, every $\Delta=\Delta'+\Delta''$ has to act only as $\Delta''$.
	Thus 
	$$P(\Delta \Htpy)^n(e^{\Sint/\hbar}I(g)) = P(\Delta'' \Htpy)^n(e^{\Sint/\hbar}I(g)) = P[((\Delta'' \Htpy)^ne^{\Sint/\hbar})I(g)],$$
	where the last holds because $I(g)$ has no variables $\beta$ or $\gamma$, so $\Delta'' \Htpy$ does not act on it and it does not affect the normalization of $\Htpy$.
	We obtain
	\begin{gather*}
	Z(Ig) = e^{-W/\hbar}\sum_{n\geq 0}P\left(\left((\hbar\Delta \Htpy)^ne^{\Sint/\hbar}\right)I(g)\right) = \\
	= (e^{-W/\hbar}P(1-\hbar\Delta \Htpy)^{-1}e^{\Sint/\hbar})\cdot PI(g) = e^{-W/\hbar}e^{W/\hbar}g = g.
	\end{gather*}
\end{proof}

\begin{proof}[Proof of Theorem~\ref{LEMEqualityOfProjections}]
	Let's evaluate
	$$I_2P_2-\id=\Htpy_2D_2+D_2\Htpy_2$$
	on $f\in \Fun{V}$ and apply $Z$ on both sides. This gives
	$$ZI_2P_2(f)-Z(f) = Z\Htpy_2D_2(f)+ZD_2\Htpy_2(f)\,.$$
	
	Here, $Z\Htpy_2D_2(f)=0$ since $\Htpy$ adds one $\beta$, and hence $\Htpy_2=\Htpy(1-\hbar\Delta \Htpy-\{\Sint,\Htpy\})^{-1}$ too, adds $\beta$, and the result is annihilated by $Z$ due to Lemma~\ref{LEMClaimOne}.
	
	Moreover, $ZI_2P_2(f)=ZIP_2(f)$ since $I_2=I + \Htpy(1-\delta_2 \Htpy)^{-1} \delta_2 I$ and we use the same argument about adding $\beta$ by $\Htpy$ and Lemma~\ref{LEMClaimOne}.
	Using Lemma~\ref{LEMClaimTwo}, we then have
	\begin{gather} \label{EQOneRandom}
	P_2(f)-Z(f) = ZD_2\Htpy_2(f).
	\end{gather}
	To deal with the RHS, we study the expression $ZD_2(f)$:
	\begin{align*}
	Z D_2(f) &= e^{-W/\hbar} P_1[ (Qf+ \hbar\Delta f + \{\Sint, f\})e^{\Sint/\hbar} ]
	\\&= e^{-W/\hbar} P_1[ (Q+ \hbar\Delta  )(fe^{\Sint/\hbar}) ]
	\\&=  e^{-W/\hbar} \hbar \Delta' P_1[ fe^{\Sint/\hbar}]\,,
	\\&=  e^{-W/\hbar} \hbar \Delta' (e^{W/\hbar} Z(f))\,.
	\end{align*}
	Since $\Htpy_2$ always adds at least one $\beta$, we have $ZD_2\Htpy_2(f) \propto \hbar \Delta' (e^{W/\hbar}Z(\Htpy_2(f) ))=0$ by Lemma~\ref{LEMClaimOne}. Equation~\eqref{EQOneRandom} thus gives $Z(f) = P'(f)$.
\end{proof}


Since the perturbation by $\delta_2$ can be obtained from perturbation by $\delta_1$ by twisting with $e^{\Sint/\hbar}$, we expect
that the perturbed differential $E_2$ on $\Fun{H}$ is a twist of $E_1 = \hbar \Delta'$, as in beginning of Section~\ref{SECBVandQL}.
\begin{theorem} \label{THMMain}
	$$E_2 = \hbar \Delta' +\{W,-\}'\,.$$
\end{theorem}
\begin{proof}
	In the proof of Theorem~\ref{LEMEqualityOfProjections}, we showed that
	$$	Z D_2(f) =  e^{-W/\hbar} \hbar \Delta' (e^{W/\hbar} Z(f)) = (\hbar\Delta' + \{W, -\}') Z(f)\,.
	$$
	Since $Z = P_2$, by the perturbation lemma we have
	$$Z D_2 = P_2 D_2 = E_2 P_2$$
	and so 
	$$E_2 P_2 = (\hbar\Delta' + \{W, -\}') P_2\,.$$
	This finishes the proof, since $P_2$ is surjective: by perturbation lemma, $I_2$ is its right inverse.
\end{proof}
\begin{remark}
	This theorem gives another formula for $W$: In the expansion $E_2 = P \delta_2 I + P \delta_2 \Htpy \delta_2 I  + P \delta_2 \Htpy \delta_2 \Htpy \delta_2 I $, the rightmost $\delta_2$ can only act by primed $\Delta$ and bracket and all the other must act by double-primed $\Delta$ and $\{, \}$, to remove $\beta$ that is added by $\Htpy$. The operator $E_2$ is thus equal to $\hbar \Delta' + X$, where $X$ is a vector field. The condition $(E_2)^2 =0$ implies the vector field $X$ is integrable to the form $\{W, -\}'$, where 
	\begin{equation*}
	W = \sum_{k=0}^\infty \frac{1}{\#_\alpha} \circ P \circ (\delta_2 \circ \Htpy)^k \circ \#_\alpha (S)\,.
	\end{equation*}
	Here, $\#_\alpha$ multiplies a monomial by the number of variables $\alpha$ in it. This approach was used in J.P.'s diploma thesis \cite{JPthesis}.
\end{remark}

\subsection{Homotopies}\label{SSECHomotopies}
We will begin by introducing homotopies of quantum $L_\infty$ algebras, following \cite{BraunMaunder, Chuang2010}. Then, we will see that the perturbation lemma directly gives a homotopy between 
the original and the effective action.

Homotopy between two solutions of quantum master equation should interpolate between them. To talk about time dependence, we tensor our space $\Funpos{V}$ with the cdga $\Omega([0, 1])$, the de Rham complex of an interval.
\begin{definition}
	By $\Omega([0, 1])$, we mean the algebra of smooth differential forms on the unit interval $[0, 1]$. Elements of this algebra can be written as $f(t) + g(t)\dd t$, the differential
	$\dd_\mathrm{dR}$ sends such element to $\partial_t f(t) \dd t$.
	
	 \medskip
	 
	 The tensor product $\Funpos{V}\otimes \Omega([0, 1])$ is defined as
	 	$$ \Funpos{V}\otimes \Omega([0, 1]) \equiv \prod_{w  \ge 1} \left[ \left( \bigoplus_{2g+n = w} \hbar^{g} (\dual{V})^{\odot n} \right)\otimes \Omega([0, 1] ) \right] \,, $$
	 	i.e. in each weight, we have coefficients given by differential forms. Since $\fld$ is reals or complex numbers, we can always set $t$ to a number between $0$ and $1$. 
\end{definition}
\begin{remark}
	Taking exponentials and logarithms of elements of $\Funpos{V}\otimes \Omega([0, 1])$ is a well defined operation, since there are only finitely many contributions to each weight. 
	Thus, in each weight, we sum finite number of finite powers of smooth functions of $t$.
	
	We will also define a convex combination as follows
		$$e^{A(t)/\hbar} = (1-t)e^{S_0/\hbar} + t e^{S_1/\hbar}\,.$$
	Here, $A(t)$ is again well defined, because the right hand side starts with $1$ and then contains terms in higher weight which are smooth (linear in fact) in $t$.
\end{remark}
A solution of the QME is given by $e^{S/\hbar}$ closed under $ \{ \Sfree, - \} + \hbar \Delta$. We will thus define homotopy as a degree zero element of $\Funpos{V}\otimes \Omega([0, 1])$, closed under the differential $\{ \Sfree, - \} + \hbar\Delta + \dd_\mathrm{dR}$.
\begin{definition}
	\label{DEFhomotopy}
	We say that $e^{(A(t) + B(t)\dd t)/\hbar}\in \Funpos{V}\otimes \Omega([0, 1])$ is a \textbf{homotopy} between $A(0)$ and $A(1)$ if $A(t)$ is of degree 0, $B(t)$ is of degree -1 and 
	\begin{equation}
	\label{EQhomotopy}
	(\{ \Sfree, - \} + \hbar\Delta + \dd_\mathrm{dR})\left(e^{(A(t) + B(t)\dd t)/\hbar}\right) = 0\,.
	\end{equation}
	This is equivalent to saying that $A(t)$ solves the quantum master equation for every $t$ and that
	\begin{equation}\label{EQHom2}
	\frac{\dd A(t)}{\dd t} + \{ \Sfree, B(t) \} + \{A(t), B(t)\} + \hbar \Delta B(t) = 0\,.
	\end{equation}
\end{definition} 
Costello \cite{Costello2010} shows that such homotopy is equivalent to a symplectic diffeomorphism $\Phi = \Phi(1) : V \to V$ given by the flow of the vector field $X(t) = -\{B(t), -\}$. 

\begin{remark}\label{RKss}
This is a variant of the well-known
correspondence between homotopies and gauge equivalences of Schlessinger and Stasheff \cite{Schlessinger}. In this case, the relevant dgla is $\sus\Funpos{V}$ with $\hbar\Delta + \{\Sfree, - \}$ as a differential and $\{-,-\}$ 
as a bracket. The Maurer-Cartan elements in this dgla are just solutions to the quantum master equation, see Section 6 of \cite{BraunLazarev}.
\end{remark}

There is also another characterization of homotopy, related to the Moser Lemma, which says that $S_0$ and $S_1$ are homotopic iff there the difference $e^{S_0/\hbar} - e^{S_1/\hbar}$ is $(\{ \Sfree, - \}+\hbar\Delta)$ exact. 
\begin{theorem}
	\label{THMhomotopy}
	Let us take two actions $S_0, S_1 \in \Funpos{V}$. Then the following three claims are equivalent:
	\begin{enumerate}
		\item There exists $F\in\Funpos{V}$ such that $e^{S_0/\hbar} - e^{S_1/\hbar} = (\{ \Sfree, - \}+\hbar\Delta)F$
		\item There exists a homotopy in the sense of Definition~\ref{DEFhomotopy} connecting $S_0$ and $S_1$
		\item There is a symplectic diffeomorphism $\Phi$ of $V$, of the form $\id + (\text{terms of positive weight})$, such that 
		\begin{equation} \label{EQHomotopy3}
		e^{(S_\mathrm{free} + S_0)/\hbar}\dd^{\frac 12} V = \Phi^*(e^{(S_\mathrm{free} + S_1)/\hbar}\dd^{\frac 12} V)\,.
		\end{equation}
	\end{enumerate}
\end{theorem}
\begin{proof}
	The equivalence of the second and the third claim is from Costello, see Section 10.1 of \cite[Chapter~5]{Costello2010}.
	\medskip 
%
	
	The implication $2. \implies 1. $ is simple, since Equation~\eqref{EQhomotopy} says that 
	$$\frac{\partial}{\partial t} e^{A(t)/\hbar} = - (\{ \Sfree, - \} + \hbar\Delta) (e^{A(t)/\hbar} B(t)/\hbar)\,,$$
	i.e. the change of $e^{A(t)/\hbar}$ is $(\hbar\Delta+\{ \Sfree, - \})$-exact.
	\medskip
	
	To show $1. \implies 3.$, we define 
	$$e^{A(t)/\hbar} \equiv (1-t)e^{S_0/\hbar} + t e^{S_1/\hbar}$$
	and consider a half-density $\mu(t) \equiv e^{(S_\mathrm{free} + A(t))/\hbar}\dd^{\frac 12} V$. Now, let us compute the time derivative of $\mu(t)$
	\begin{align*}
	\dot{\mu}(t) &=  (e^{S_1/\hbar} - e^{S_0/\hbar})e^{\Sfree/\hbar}\dd^{\frac 12} V= - (\{ \Sfree, F \}+\hbar\Delta F ) e^{\Sfree/\hbar}\dd^{\frac 12} V \\ &= - (\hbar\Delta + \{ \Sfree, - \} + \{A(t) , -\})(Fe^{-A(t)/\hbar}) \mu(t)\,,
	\end{align*}
	where we used Equation~\eqref{EQTwistDelta} and the fact that $A(t)$ also satisfies the QME. Last step is using the following version of Equation~\eqref{EQorigdivergence}
	\begin{equation}\label{EQdivergence}
	(\hbar \Delta f + \{G, f\} ) e^{G/\hbar}\dd^{\frac 12} V = (-1)^{\hdeg f} \hbar \mathcal{L}_{\{f, - \}}(e^{G/\hbar}\dd^{\frac 12} V)\,,
	\end{equation}
	for $G\in\Fun{V}$ which is a solution of the quantum master equation.
    Using this, we can write the time derivative of $\mu(t)$ as 
	$$\dot{\mu}(t) = -\mathcal{L}_{\{ -\hbar Fe^{-A(t)/\hbar}, -\}} \mu(t)\,,$$
	i.e. $\mu(t)$ is given by a $\mu(t) = (\Phi_t)_*\mu(0)$, where $\Phi_t$  is the flow of the vector field $\hbar\{- Fe^{-A(t)/\hbar}, -\}$. For $t=1$, we get exactly the claim~3. The homotopy in the sense of Definition~\ref{DEFhomotopy} is explicitly given by 
	$$e^{A(t)/\hbar} + F \dd t \,.$$
\end{proof}
\begin{remark}
	The first condition of Theorem~\ref{THMhomotopy} can be rewritten as 
	$$ e^{(S_\mathrm{free}+S_0)/\hbar} - e^{(S_\mathrm{free}+S_1)/\hbar} = \hbar\Delta (Fe^{S_\mathrm{free}/\hbar})\,.$$
	Multiplying with the volume form $\dd V$ and using $2\Delta f\dd V = (-1)^{\hdeg f}\mathcal{L}_{\{f, - \}}\dd V$ we can write
	$$ e^{(S_\mathrm{free}+S_0)/\hbar}\dd V - e^{(S_\mathrm{free}+S_1)/\hbar}\dd V = -\frac\hbar 2 \dd \left( i_{\{Fe^{S_\mathrm{free}/\hbar},-\}}\dd V\right)\,.$$
	 The above equation then just says that $e^{(S_\mathrm{free}+S_0)/\hbar}\dd V$ and $e^{(S_\mathrm{free}+S_1)/\hbar} \dd V$ lie in the same cohomology class. Thus, the fact that these volume forms are connected by a homotopy is a (graded version) of the Moser Lemma \cite{moser1965volume}.
\end{remark}
\begin{remark} \label{RMKUniqueness}
	From this theorem, one can easily see that homotopic solutions of QME on $V$ integrate to
	homotopic effective actions: if  $e^{S_0/\hbar} - e^{S_1/\hbar} = (\{ \Sfree, - \}+\hbar\Delta)F$, the difference of the effective actions is given by $P_1 (\{ \Sfree, - \}+\hbar\Delta)F = \hbar\Delta'\,P_1 (F)$, which gives a homotopy in $\Fun{\Hlgy}$. Similarly, one can show that two actions which give the same effective action (up to a $\hbar\Delta'$-exact term) are homotopic.
\end{remark}

\subsubsection{Constructing a homotopy between $e^{W/\hbar}$ and $e^{S/\hbar}$}
Now, we would like to find a homotopy between the original and the effective action. Recall that from the SDR obtained after perturbation by $\hbar\Delta$, we have 
$$I_1 P_1 (e^{\Sint/\hbar}) - e^{\Sint/\hbar} = Q_1 \Htpy_1 e^{\Sint/\hbar}  + \Htpy_1 Q_1 e^{\Sint/\hbar}\,.$$
Remembering that $Q_1 = \{ \Sfree, - \} + \hbar\Delta$, $I_1 P_1(e^{\Sint/\hbar} ) = e^{I(W)/\hbar}$ and that $e^{\Sint/\hbar} $ is $Q_1$-closed, we obtain 
$$e^{I(W)/\hbar} - e^{\Sint/\hbar} = Q_1 \Htpy_1(e^{\Sint/\hbar} )\,.$$
Now we can use Theorem~\ref{THMhomotopy} to find a homotopy between these two solutions of the QME: the flow between these two actions is given by the vector field
$$X(t) = -\hbar\{ e^{-A(t)/\hbar} \Htpy_1e^{\Sint/\hbar} , -\}\,.$$
\begin{remark}
	This amounts to a special choice of $F = \Htpy_1(e^{\Sint/\hbar} )$. It is however, a natural one: out of all possible such $F$, it is the one that satisfies $P_1(F) = 0$ and $\Htpy_1(F) = 0$. In other words, because $1 = I_1P_1 - \Htpy_1Q_1 - Q_1\Htpy_1$, we chose an $F$ in the image of the projector $-\Htpy_1Q_1$.
\end{remark}
One can, for example, integrate this flow using the Magnus expansion, which will give us an answer
in the form $\Phi_t = \exp( \{M(t), -\} )$, for degree -1 element $M(t)\in \Funpos{V}\otimes \Omega([0, 1])$ (see Section 3.4.1 in \cite{MagnusBCOR}). The first term of the expansion is 
$$M(t) = \frac{1}{e^{W/\hbar} - e^{\Sint/\hbar}}\hbar \log\left[ 1 + (e^{(W-\Sint)/\hbar} - 1)t \right] \Htpy_1(e^{\Sint/\hbar}) + \dots \,.$$
\begin{remark}
	This linear interpolation works for any standard situation: there is a chain map  $V \to V\otimes \Omega([0, 1])$, given by
	\begin{equation*}
	v \mapsto  (1-t) v + t\, ip(v) - (-1)^{\hdeg{v}} \htpy(v) \dd t\,,
	\end{equation*}
	where we use the notation from Definition~\ref{DEFSS}.
	This map therefore gives a homotopy between $v$ and $ip(v)$ for every closed $v$.
\end{remark}
\subsection{Morphisms}\label{SSECMorphism}
The correct notion of morphisms of quantum $L_\infty$ algebras should come from Lagrangian correspondences (see \cite[remark~2.4.6]{gwilliam-thesis}). 
However, we can define a more restrictive notion, a Poisson map preserving the differentials $T_S$.
\begin{definition}
Given two symplectic vector spaces $(U, \omega_U)$, $(V, \omega_V)$ and solutions of master equation  $S_U \in\Fun{U}$, $S_V \in \Fun{V}$,
we say that a formal map (fixing the origin) $\Phi: U\to V$ is a \textbf{quantum $L_\infty$-morphism} if 
$$\Phi_*(\omega^{-1}_U) = \omega^{-1}_V$$
i.e. if it's a Poisson map, and if 
$$\Phi^* \circ T_{S_V} = T_{S_U} \circ \Phi^*\,,$$
i.e.
\begin{equation}\label{EQMorphdef}
\Phi^* \circ (\hbar\Delta_V f + \{S_V, f\}_V) = \hbar \Delta_U \Phi^* f + \{S_U, \Phi^* f  \}_U\,,
\end{equation}
for any $f \in \Fun{V}$.
\end{definition}
Note that since $\Phi$ is a Poisson map, we have $\dim U \ge \dim V$.

If a quantum $L_\infty$ morphism is a symplectic automorphism of $(V, \omega)$,
it corresponds to a homotopy in the sense of Definition~\ref{DEFhomotopy}.
\begin{lemma}\label{LEMHomMorph}
	A homotopy between two actions $S_0$ and $S_1$, given as a symplectic diffeomorphism $\Phi$, is a quantum $L_\infty$
	isomorphism $\Phi: (V, \omega, \Sfree + S_0) \to (V, \omega, \Sfree + S_1)$. 
	\medskip
	
	Conversely, a quantum $L_\infty$ isomorphism
	$\Phi: (V, \omega, S_0) \to (V, \omega, S_1)$ of the form $\Phi = 1 + (\text{terms of positive weight})$ 
	is also a homotopy between $\Sfree+S_0+C$ and $\Sfree+S_1$, for some $C\in\Funpos{V}$ of polynomial degree 0.
%
%
\end{lemma}
\begin{proof}
	Taking a logarithm of Equation~\eqref{EQHomotopy3}, we get 
	\begin{equation}\label{EQlogmorph}
	\Phi_*(\Sfree + S_0) + \frac 12 \hbar\log \Ber(\Phi) = \Sfree + S_1\,. 
	\end{equation} 
	Using Equation~\ref{EQBVTransformation}, this implies
	\[ \Phi_*  \{ \Sfree +S_0, - \} \Phi^* + \Phi_* \hbar\Delta \Phi^*  = \hbar\Delta + \{\Sfree + S_1 , -\} \,.\]
	For an opposite implication, because the symplectic form is non-degenerate, its center consists of constants. Thus, Equation~\ref{EQMorphdef} implies Equation~\ref{EQlogmorph} up to a constant.
\end{proof}
In the previous section, we have constructed a symplectic diffeomorphism $\Phi \colon V \to V$ which satisfies
$$e^{(S_\mathrm{free} + I(W))/\hbar}\dd^\frac12 V = \Phi^*(e^{(S_\mathrm{free} + \Sint)/\hbar}\dd^\frac12 V)\,. $$
Thus, have from Lemma~\ref{LEMHomMorph} that $\Phi$ is a 
quantum $L_\infty$-morphism $(V, \omega, \Sfree + I(W)) \to (V, \omega, \Sfree + \Sint )$\,.

Now, recall that $p \colon V \to H$ pulls back to the map $I$, i.e. $I(W) =p^*(W)$. It is,
however, also a quantum $L_\infty$ morphism, and we can compose it with $\Phi^{-1}$
to a quantum $L_\infty$ morphism from $(V, \omega, \Sfree + \Sint)$ to the homology.
\begin{lemma}
	Given a quantum $L_\infty$ algebra $(V, \omega, S_1)$, there is a $L_\infty$ morphism from the original 
	quantum $L_\infty$ algebra to its minimal model
	$$p \circ \Phi^{-1} \colon (V, \omega, \Sfree + \Sint) \to (H, \omega_H, W)\,.$$	
\end{lemma}
\begin{proof}
	The projection $p$ is a Poisson map, since
	$\{I(G_1), I(G_2)\} = I(\{G_1, G_2\}')$ for $G_i \in \Fun{H}$ and also satisfies Equation~\eqref{EQMorphdef}, since
	$$
	T_{\Sfree + I(W)} (I(G)) = \hbar \Delta I(G) + \{\Sfree + I(W), I(G)\} = I(\hbar \Delta' G + \{W, G\}')\,,
	$$
	or
	$$T_{\Sfree + I(W)} \circ p^* = p^* \circ T'_{W}\,.$$
\end{proof}
\section{Related works}
\label{SECRemarks}
The connection of the homological perturbation lemma and path integrals is known among experts. It appears most explicitly in a lecture by Carlo Albert \cite{Albert}, but see also remarks by Costello \cite[Chapter 5, Section 2.7]{Costello2010}, Cattaneo, Mnev and Reshetikhin \cite[Theorem 8.1]{CMR}, a paper by Gwilliam and Johnson-Freyd \cite[remark in Section 3]{TJF}. The most detailed reference is an example worked out by Gwilliam in his thesis \cite[Section~2.5]{gwilliam-thesis}, see also the respective subsection. In this section, we explain how our work fits with their.

\subsection{Kajiura}
Kajiura \cite{Kajiura} considers a classical, associative case, the \emph{cyclic $A_\infty$ algebra}. He proves a decomposition theorem, constructing a \emph{cyclic $A_\infty$-isomorphism} between the original algebra and a direct sum of a minimal and a linear contractible $A_\infty$ algebras. The linear contractible algebra contains only the differential and the minimal one has a zero differential, but contains all the higher brackets of the minimal model. The minimal model is constructed iteratively (reminiscing the homological perturbation lemma), giving sums over trees as a result. Our decomposition of the action and the homotopy between $\Sfree + \Sint$ and $\Sfree + W$ is an analogue of this construction in the quantum BV formalism. 

\subsection{Mnev}
Mnev \cite{Mnev2008} defines an effective action using the path integral in the BV formalism. He also shows that small deformations of the Hodge decomposition change the effective action by a \emph{canonical transformation} $W \to W + \{W, R\} + \hbar\Delta R$, which is an infinitesimal version of the usual homotopy from Definition~\ref{DEFhomotopy}.

Mnev also interprets the action as a quantum $L_\infty$ algebra. His BF theory is constructed from a dgla $V_0$ by setting $V = V_0[1] \oplus \dual{V_0}[-2]$, since $V$ then has a canonical odd symplectic structure (the pairing is then of degree 1 in his convention). The dgla is extended onto $V$, the classical master equation is true and $\hbar\Delta S = 0$ iff the original dgla is \emph{unimodular} (the supertrace of the adjoint representation  is zero). Because of this special structure of $V$ (considered also by Barannikov, in the associative case), the Feynman diagrams of the expansion are oriented and there is only a trivalent vertex, with two incoming and one outcoming edge. In this case, the diagrams can only have up to one loop, which means that the effective action has only zeroth and first powers of $\hbar$. 

Mnev calls this first-order action a quantum $L_\infty$ algebra, but it has later been called a \emph{unimodular $L_\infty$ algebra} in a related work of Gran{\aa}ker \cite{Granaker2008}, who interprets the effective action as a minimal model.

\subsection{Costello \& Gwilliam}
In the finite-dimensional case, Costello's propagator $P(0, \infty)$ (see \cite[Section 6.5]{CostelloRBV} or \cite[Chapter 2, Sections 3,4]{Costello2010}) is equal to our propagator. However, Costello defines the Feynman diagrams without the projection and for general propagator $P(\varepsilon, L)$, which in our case would not work -- the exponential $\exp(\hbar\partial_P)$ in Theorem~\ref{THMCertainComputationII} can be reconstructed only if we apply the projection. It would be interesting to see whether one can modify the HPL input data to obtain exponentials in general.

Gwilliam in his thesis \cite{gwilliam-thesis} gives an example of how the HPL gives the Feynman expansion when perturbing by $\hbar\Delta$, as well as constructing the perturbation retract on $\Fun{V}$. This is identical to our Theorem~\ref{THMCertainComputationII} and the preceding construction.

\subsection{Chuang \& Lazarev, Braun \& Maunder}
Chuang  and Lazarev \cite{Chuang2010} construct a minimal model and a homotopy equivalence for any modular operad. The  minimal model is given by a sum over all stable graphs, with propagators given by homotopy $s$ and the form on $V$. In future, we would like to understand the relation of their approach to homotopy to ours.

Braun and Maunder \cite{BraunMaunder} define the path integral explicitly and use it to compute the effective action. They then prove that the effective action again solves a quantum master equation (and hence defines a quantum $L_\infty$ algebra). Moreover, they show that the homotopy classes of quantum $L_\infty$ algebras on $V$ and its cohomology are in bijection and that (in our language) $I(W)$ is homotopic to $\Sint$.

Their path integral coincides with our map $P_1$, which can be seen from the Wick Lemma \cite[Theorem~A.6]{BraunMaunder}: the integral of a monomial is given by a sum over all possible pairings. The propagator is given by the inverse of $\sigma = \langle -, \dd - \rangle$, where $\langle,\rangle$ is their odd symplectic form and $\dd$ is the differential. This is, up to sign conventions, the propagator in our Theorem~\ref{THMCertainComputationII}.

\subsection{M\"unster \& Sachs}
M\"unster and Sachs prove a decomposition theorem in \cite{MuensterSachsClassification} for quantum $L_\infty$ algebras, again by defining it by the Feynman expansion. Their loop homotopy algebra is the same as our quantum $L_\infty$ algebra, but they work in a category of $IBL_\infty$ algebras, which is bigger. They also describe a flow between two quantum $L_\infty$ algebras and use it to show the uniqueness of closed string field theory. This argument, in our language, is contained in Remark~\ref{RMKUniqueness}.

\subsection{Barannikov}
In \cite[Section 4]{BarannikovSolving}, Barannikov gives a general formula for transferring solutions of QME, for any modular operad. For the modular extension of the $L_\infty$ operad, these correspond to the formulas from Theorem~\ref{THMCertainComputationII}. Specifically, the propagator is a composition of the dual scalar product and the homotopy. 

The sum in \cite{BarannikovSolving} is over \emph{stable} graphs, i.e. graphs for which every vertex $v$ has an assigned number $b(v)$ and $2b(v) + n(v) - 2 > 0$, where $n(v)$ is the number of edges adjacent to the vertex. In the graph sum, $b(v)$ corresponds to the power of $\hbar$ and $n(v)$ to the polynomial degree, so the condition $2b(v) + n(v) - 2 > 0$ means we consider only vertices with \emph{weight grading} bigger than 2, which is our condition on $\Sint$.

\subsection*{Acknowledgments}
The research of M.D. and B.J. was supported by grant GAČR P201/12/G028.  B.J. wants to thank MPIM in Bonn for hospitality. J.P. was supported by NCCR
SwissMAP of the Swiss National Science Foundation and had also benefited from support by the project SVV-260089 of the Charles University. We would like to thank the anonymous reviewer, for suggestions  that streamlined the paper considerably, and also for explaining to us what is now Remark~\ref{RKss}. B.J. thanks Martin Markl and Owen Gwilliam for discussions. J.P would like to thank Florian Naef and Pavol Ševera for numerous discussions and Pavel Mnev for useful pointers. Finally, we would like to thank Lada Peksová for her useful comments on an earlier draft of this paper.


\providecommand{\href}[2]{#2}\begingroup\raggedright\endgroup

\end{document}